\documentclass[preprint,12pt]{article}

\usepackage{authblk}
\usepackage{amssymb, amsfonts}
\usepackage{amsthm}
\usepackage{subcaption}
\usepackage{multirow}
\usepackage{mathtools}
\usepackage{cuted,url}
\usepackage{booktabs,array,listings}

\newtheorem{theorem}{Theorem}

\newtheorem{corollary}{Corollary}[theorem]

\begin{document}

\title{Strong Sleptsov Net is Turing-Complete}
\author{Dmitry~A.~Zaitsev\footnote{ORCID: 0000-0001-5698-7324, e-mail: daze\emph{@}acm.org; web-site: daze.ho.ua; University of Information Technology and Management in Rzeszów, ul. Sucharskiego 2, Rzeszow, 35-225, Poland.}}
\date{}
\maketitle

\begin{abstract}
It is known that a Sleptsov net, with multiple firing a transition at a step, runs exponentially faster than a Petri net opening prospects for its application as a graphical language of concurrent programming. We provide classification of place-transition nets based on firability rules considering general definitions and their strong and weak variants. We introduce and study a strong Sleptsov net, where a transition with the maximal firing multiplicity fires at a step, and prove that it is Turing-complete. We follow the proof pattern of Peterson applied to prove that an inhibitor Petri net is Turing-complete simulating a Shepherdson and Sturgis register machine. The central construct of our proof is a strong Sleptsov net that checks whether a register value (place marking) equals zero.

\end{abstract}

\section{Introduction}
\label{}
A timed Petri net with multichannel transitions \cite{ZaitsevPhD,CiSA97} was a forerunner of a uniform speed-up technique, often called ``an exhaustive use of rule'', for many modern universal computation systems including spiking neuron systems \cite{Neary,ExUoR1,ExUoR2}, DNA computing \cite{DNAcomp,Winfree1,Winfree2,Winfree3}, and multiset rewriting systems \cite{Alhazov}. The technique application allowed researchers to obtain exponential speed-up of computations. The corresponding class of untimed place-transition nets was called a Sleptsov net \cite{SNRF}. It seems rather reasonable to classify place-transition nets with regard to their transition firing rules: a Petri net \cite{Petri62,Peterson81} -- one transition fires at a step; a Salwicki\footnote{Hans-Dieter Burkhard told me in personal communication that Andrzej Salwicki hinted him the idea before writing and publishing \cite{Burkhard}.} net \cite{Burkhard} -- the maximal set of firable transitions fires at a step; a Sleptsov\footnote{Anatoly Sleptsov hinted me the idea in 1988 when I entered PhD course \cite{ZaitsevPhD} under his supervision.} net \cite{SNRF} -- multiple firing of a transition at a step. Sleptsov nets are applied as a general-purpose graphical language of concurrent programming \cite{SNC}. Sometimes it is useful to combine Salwicki and Sleptsov nets to compose fast parallel processes. 

From the computation point of view, a Petri net is known as more powerful than a state machine (finite automaton) and less powerful than a Turing machine. Such amendments to a Petri net as an inhibitor arc \cite{Tilak,Hack74} or priority of transitions \cite{Hack76} make the model Turing-complete. That is why for computational purposes Sleptsov nets with inhibitor arcs and priorities have been applied \cite{SNRF,PSNL,USN}. From the other hand, Burkhard has proven that Salwicki net (without any amendments) is Turing-complete \cite{Burkhard}. It was also proven that infinite Petri net is Turing-complete \cite{UIPN}. As a collateral answer to Shannon contest \cite{Shannon} for a small universal Turing machine, a series of universal constructs \cite{Korec,Neary}, including universal Petri and Sleptsov nets \cite{UPN,TMUPN,SCAIPN,SUDPN, USN}, has been presented. 

In the present paper, we prove that a strong Sleptsov net, where a transition with the maximal firing multiplicity fires at a step, is Turing-complete. Tilak Agerwala \cite{Tilak} has proven Turing-completeness of an inhibitor Petri net via direct simulation of a Turing machine \cite{Turing}. We follow the later, more compact, proof pattern offered by James Peterson \cite{Peterson81} via simulation of Shepherdson and Sturgis \cite{Shepherdson} register machine, close to Minsky program machine \cite{Minsky}, known as Turing-complete. As a crucial construct of the proof, a strong Sleptsov net that checks whether a register (place) equals zero has been constructed, only nonnegative integer numbers are considered.

\section{Classification of Place-Transition Nets Based on Firing Rules}
Schemata of manufacturing processes has adopted parallel processes notation much more earlier than program schemata, a standard for process charts in manufacture was issued in 1947 \cite{ASMEflow}, preliminary research results published in 1921 \cite{Gilbreth}. The first flow-chart standard of von Neumann and Goldstein \cite{FlowChart} of 1949, as well as program schemata of Yanov \cite{Yanov1,Yanov2} of 1958, did not contain facilities for specification of parallel processes. 

In 1958, Gill \cite{Gill} presented schemata of parallel computing processes; his schemes represent a bipartite graph with two parts of vertices depicted as circles and rectangles (though a circle serves for splitting and joining processes). A place-transition net introduced by Carl Petri in 1962 \cite{Petri62} represents, in essence, a unification of process nets presented by Frank and Lilian Gilbreth in 1921 \cite{Gilbreth}, only two kinds (parts) of vertices remain -- places (depicted as circles) to represent conditions and transitions (depicted as rectangles) to represent events. Introduction of a dynamic element, called a token, by Carl Petri \cite{Petri62} enriched and empowered the model. Note that recently, UML \cite{UML} has adopted place-transition notation to represent parallel processes with activity diagrams.

\subsection{Place-transition Nets} 

We define a \textit{place-transition net} (PTN) as a bipartite directed multigraph with a dynamical process (behavior) specified on it: a PTN is a tuple $N=(P,T,A,\mu_0)$, $P=\{p\}$ is a finite set of places, $T=\{t\}$ is a finite set of transitions, mapping $A: P \times T \cup T \times P \rightarrow \mathbb{Z}_{\geq 0}$ specifies arcs and their multiplicity, mapping $\mu: P \rightarrow \mathbb{Z}_{\geq 0}$ represents marking of places, $\mu_0$ denotes the initial marking. 

The net \textit{behavior} is a discrete process of changing its marking as a result of firing transitions. Discrete time is considered as a sequence of enumerated steps (tacts) $\tau=1,2,\dots$. During a tact the net marking stays unchanged and changes in the moment when the current tact changes. We denote as $\mu^\tau$ the net marking at step $\tau$. In the present paper, we study systematically and classify transition firing rules. We are interested in modifications which make the system Turing-complete.

We introduce a \textit{firability multiplicity of an arc} $(p,t)$ directed from place $p$ to transition $t$ (that implies $a(p,t)>0$) as follows: 
$$c(p,t)=\mu(p) / a(p,t),$$ 
where the division operation is considered as a whole division. Further we introduce a \textit{firability multiplicity of a transition} as follows: 
$$c(t)=min_{A(p,t)>0}(c(p,t)).$$ 
Finally, we define a \textit{firable transition} as a transition having 
$$c(t)>0.$$

Firstly, we concentrate on PTNs which do not require structural amendments using the same graph. We define: a \textit{Petri net} as a PTN where a single firable transition fires at a step; a \textit{Salwicki net} as a PTN where a maximal subset of firable transitions fires at a step; and a \textit{Sleptsov net} as a PTN where a firable transition with the maximal firability multiplicity fires in the maximal number of copies at a step. Note that, all three cases imply choice of firable transition or a subset of firable transitions which is traditionally implemented in a nondeterministic way, in case, there are no additional restrictions. Further, we apply a multiset \cite{multiset1,multiset2} to specify firable and firing transitions. For Petri and Sleptsov nets, we use a multiset containing one transition: for a Petri net, a transition belongs to the multiset with multiplicity 1, while for a Sleptsov net, it belongs with multiplicity $c(t)$. 

When a chosen multiset $F$ of transitions \textit{fires}, it decreases the places marking according to its transitions' incoming arcs and their multiplicity and it increases the places marking according to its transitions' outgoing arcs and their multiplicity in the following way for the three defined above classes of PTNs. 
For a Petri net: 
$$\mu^{\tau+1}(p)=\mu^{\tau}(p) - a(p,t) + a(t,p),~t \in F,~p \in P.$$ 
For a Salwicki net: 
$$\mu^{\tau+1}(p)=\mu^{\tau}(p) - \sum_{t \in F} a(p,t) + \sum_{t \in F} a(t,p),~p \in P.$$
For a Sleptsov net: 
$$\mu^{\tau+1}(p)=\mu^{\tau}(p) - c(t){\cdot}a(p,t) + c(t){\cdot}a(t,p),~t \in F,~p \in P.$$
Note that, the maximality of the multiset $F$ for Salvicki net is considered with regard to validity of the next marking $\mu^{\tau+1}$ which should not contain negative numbers i.e. adding a firable transition to a maximal $F$ leads to an invalid marking.

To unite advantages of a Salwicki net, that fires a few transitions, with advantages of a Sleptsov net, that fires a transition in a few copies at a step, we introduce a \textit{Salwicki-Sleptsov} net as a PTN where a maximal multiset of firable transitions fires at a step. Note that the maximality of a multiset is defined as it is not a sub-multiset of any other firable multiset of transitions which leads to a valid next marking. For a Salwicki-Sleptsov net: 
$$\mu^{\tau+1}(p)=\mu^{\tau}(p) - \sum_{t \in F} c(t,F){\cdot}a(p,t) + \sum_{t \in F} c(t,F){\cdot}a(t,p),~p \in P.$$.

Also we introduce a modification of a Sleptsov net where a transition with the maximal firing multiplicity fires at a step and call it a \textit{strong Sleptsov net}. The formula for changing the place marking of Sleptsov net remains valid, the modification concerns only the firable transition choice to create set $F$ at the current step:
$$F=\{t' | t' \in T, \forall t \in T: c(t') \geq c(t) \}.$$
Note that, the net still possesses nondeterministic behavior because there can be a few transitions with the same maximal firing multiplicity.

We illustrate difference of firing rules and their effect on the reachability graph (RG) \cite{Peterson81} of the corresponding PTN using a simple example taken from \cite{SEPNC} and extended for Salwicki and Sleptsov nets for summing two natural number without control flow. The net is shown in fig.~\ref{fig-net-ex-ptn}. Further for representing a marking, we use either vector notation $(2,3,0)$ or multiset notation $\{2 \cdot p_1,3 \cdot p_2\}$, the examples correspond to the initial marking of PTN shown in fig.~\ref{fig-net-ex-ptn}. Unit multiplicity multiplier before the place name is usually omitted. Remind that an RG \cite{Peterson81} represents a state space of PTN containing all valid markings depicted as vertices and transitions between them depicted as arcs labeled by firing transitions set $F$.

As it was proven, a Salwicki net is Turing-complete \cite{Burkhard} and a Sleptsov net runs exponentially faster than a Petri net \cite{SNRF}. The later fact opens prospects for Sleptsov net application as a uniform general purpose graphical language of concurrent programming \cite{PSNL}.

\subsection{Inhibitor and Priority Nets}

Historically, Turing-completeness has been proven at first for inhibitor \cite{Tilak,Hack74} and priority nets \cite{Hack76,Kotov84} introduced as an extension to a Petri net. An \textit{inhibitor arc}, depicted with a small circle at its end, is firable when the place marking is equal to zero. \textit{A priority relation} can be expressed as a set of arcs connecting transitions $Y \subset T \times T$. In case two transitions $t$ and $t'$ are firable and $(t,t') \in Y$, only $t$ can fire. An alternative way to introduce priorities is assigning integer numbers (values of priority) to transitions, a firable transition with the maximal priority value fires. Inhibitor and priority nets for computing sum of two numbers by adding sequentially the first and the second number are shown in fig.~\ref{fig-net-ex-ipn} and \ref{fig-net-ex-ppn}, respectively, their RG (the same for both nets) represented in fig.~\ref{fig-net-ex-rg-ippn}. It is an interesting fact that RG of strong Sleptsov net shown in fig.~\ref{fig-net-ex-rg-snmf2} looks like a condensation of RG for an inhibitor or priority net show in fig.~\ref{fig-net-ex-rg-ippn} where a sequence of firing the same transiton is replaced by a single arc with the maximal firing multiplicity. 

Inhibitor and priority arcs represent a convenient facility for practical graphical programming in Sleptsov (Petri) nets, they have been employed for composing a control flow in the form of moving zero marking (hole) and composing nets implementing basic arithmetic and logic operations for a Sleptsov net \cite{SNRF,PSNL} as well as composing a small universal Sletpsov net \cite{USN}. For convenience of definition, the firability multiplicity of an inhibitor arc in Sleptsov net is considered equal to infinity to not restrict the total firing multiplicity of a transition \cite{SNRF}.

Note that, the small example shown in fig.~\ref{fig-net-ex} does not allow us to illustrate all the peculiarities of the considered firing rules mainly because it is not enough rich to represent alternative maximal sets of firable transitions for a Salwicki net, transitions with the same firable multiplicity for a Sleptsov net, and alternative maximal firable multisets for a Salwicki-Sleptsov net. Because of the mentioned possibilities, nondeterministic choice is present in the three mentioned classes of PTNs. 

\begin{figure}
\begin{center}
\begin{subfigure}[h]{0.32\textwidth}\centering{\scalebox{0.32}{\includegraphics{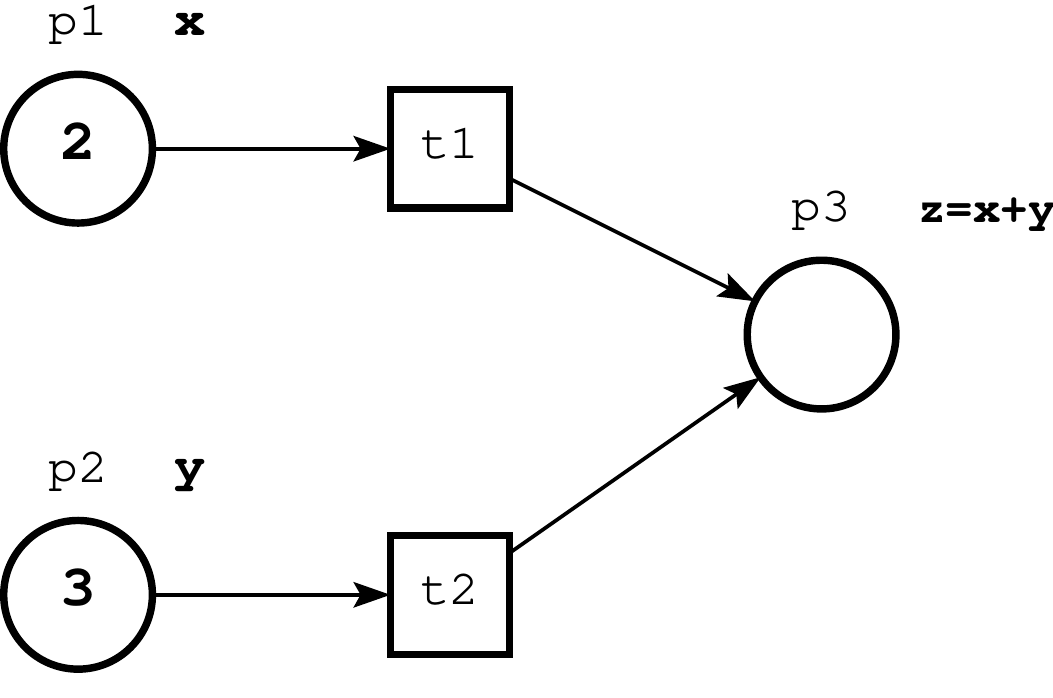}}} 
\caption{place-transition net (Petri, Salwicki, Sleptsov);}
\label{fig-net-ex-ptn}
\end{subfigure}
\begin{subfigure}[h]{0.32\textwidth}\centering{\scalebox{0.32}{\includegraphics{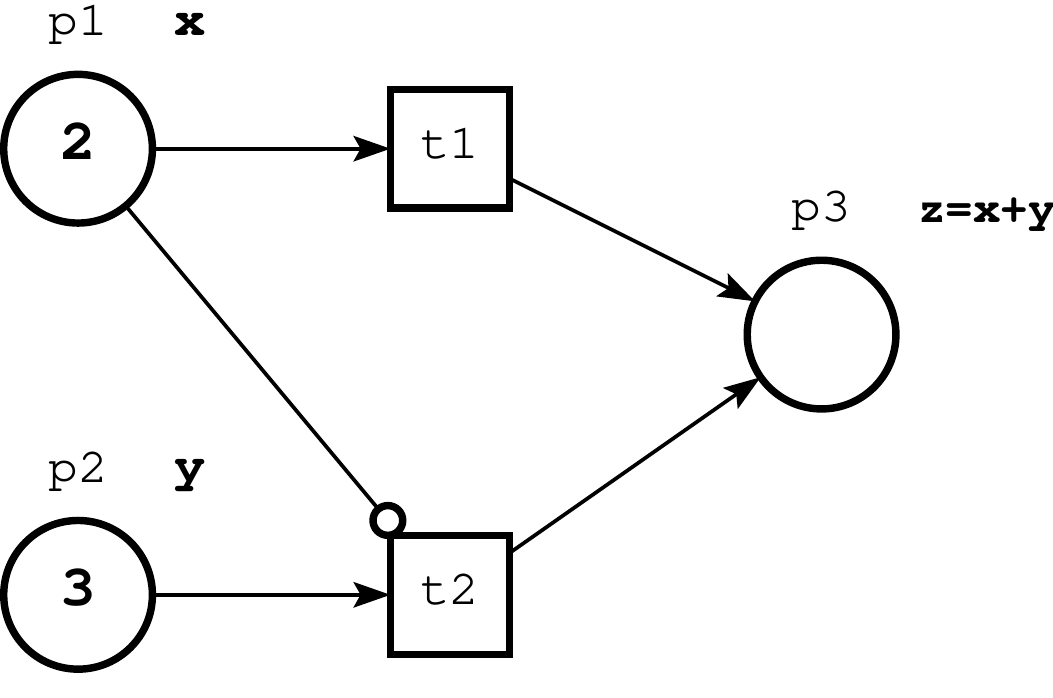}}} 
\caption{inhibitor net;}
\label{fig-net-ex-ipn}
\end{subfigure}
\begin{subfigure}[h]{0.32\textwidth}\centering{\scalebox{0.32}{\includegraphics{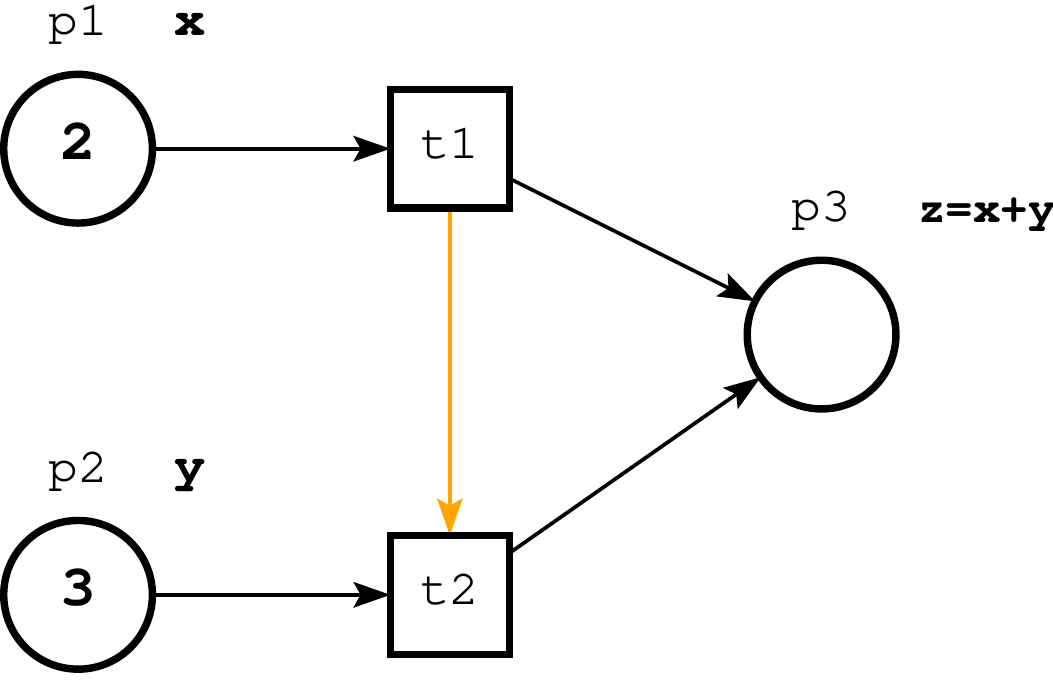}}} 
\caption{priority net.}
\label{fig-net-ex-ppn}
\end{subfigure}
\caption{Nets computing sum of two numbers. }
\label{fig-net-ex}
\end{center}
\end{figure}

\subsection{Weak and Strong Place-transition Nets}

In \cite{SNRF}, after definition of Sleptsov nets with inhibitor arcs, possibility of defining strong and weak variants of the firing rule, attractive from some practical points of view, have been mentioned. Here we explore this remark to present a regular classification of considered classes of PTNs. 

Having a \textit{general definition}, we consider as a \textit{strong (restricted) definition} its variant where a multiset with the maximal number of transitions is chosen and we consider as a \textit{weak (relaxed) definition} its variant where any valid sub-multiset of firable transitions is chosen. Let us mention, that the strong variant still allows alternatives. The described above classification is represented systematically in table~\ref{fire-class-table}.

\begin{table}[h!]
\small
\caption{Classification of PTNs.\label{fire-class-table}}
\begin{tabular}{|m{3cm}|m{3cm}|m{3cm}|m{3cm}|}
\toprule
Net type & Strong definition & General definition & Weak definition\\
\midrule
Petri &--&Any firable transition&--\\
\hline
Salwicki &The maximal number of firable transitions&A maximal set of firable transitions&A subset of firable transitions\\
\hline
Sleptsov &A transition with the maximal number of firable copies&A firable transition in the number of firable copies&A firable transition in any number of copies equal to or less than the number of firable copies\\
\hline
Salwicki-Sleptsov &The maximal total number of firable transitions copies&A maximal multiset of firable transitions&Any sub-multiset of firable transitions\\
\bottomrule
\end{tabular}\\[10pt]
\end{table}

So far as a Petri net fires a single transition, strong and weak variants coincide with the general one for it. 
For a Salwicki net, a subset containing the maximal number of transitions fires in the strong variant and any subset of firables transitions can fire in the weak variant. 
For a Sleptsov net, a firable transition with the maximal firability multiplicity fires in the strong variant and any firable transition fires in the number of copies less than or equal to the firability multiplicity.
For a Salwicki-Sleptsov net, a maximal sub-multiset containing the maximal number of transition copies fires in the strong variant and any sub-multiset of the firable transitions multiset fires in the weak variant.  

\begin{figure}
\begin{center}
\begin{subfigure}[h]{0.32\textwidth}\centering{\scalebox{0.32}{\includegraphics{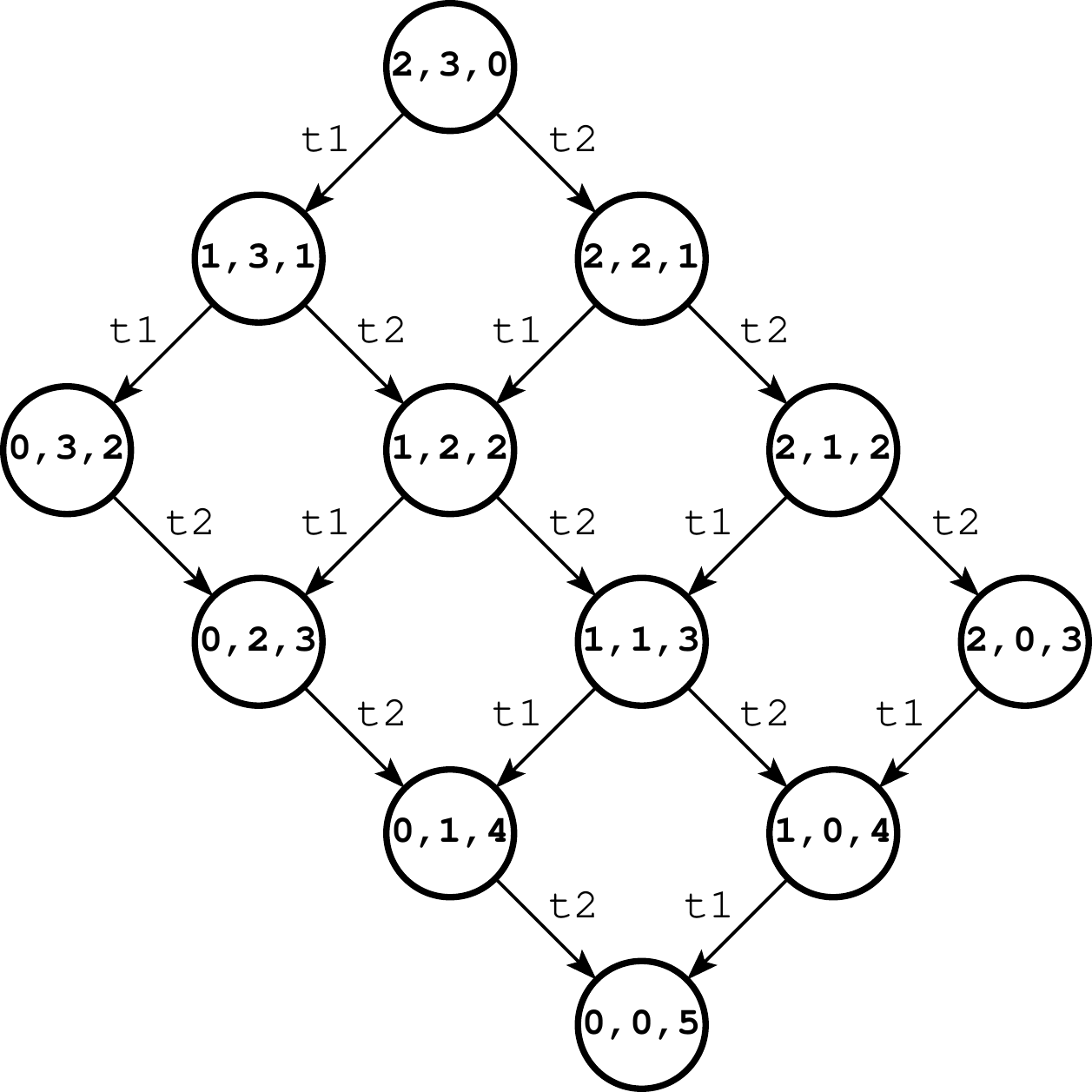}}} 
\label{fig-net-ex-rg-pn}
\caption{Petri net;}
\end{subfigure}
\begin{subfigure}[h]{0.32\textwidth}\centering{\scalebox{0.32}{\includegraphics{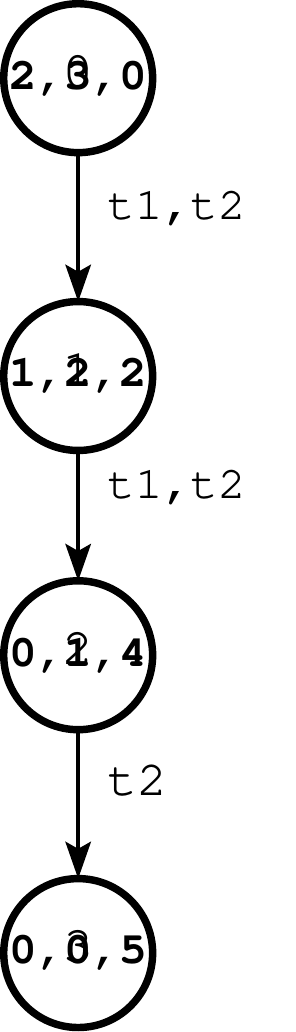}}} 
\caption{Salwicki net.}
\label{fig-net-ex-rg-wn}
\end{subfigure}
\begin{subfigure}[h]{0.32\textwidth}\centering{\scalebox{0.32}{\includegraphics{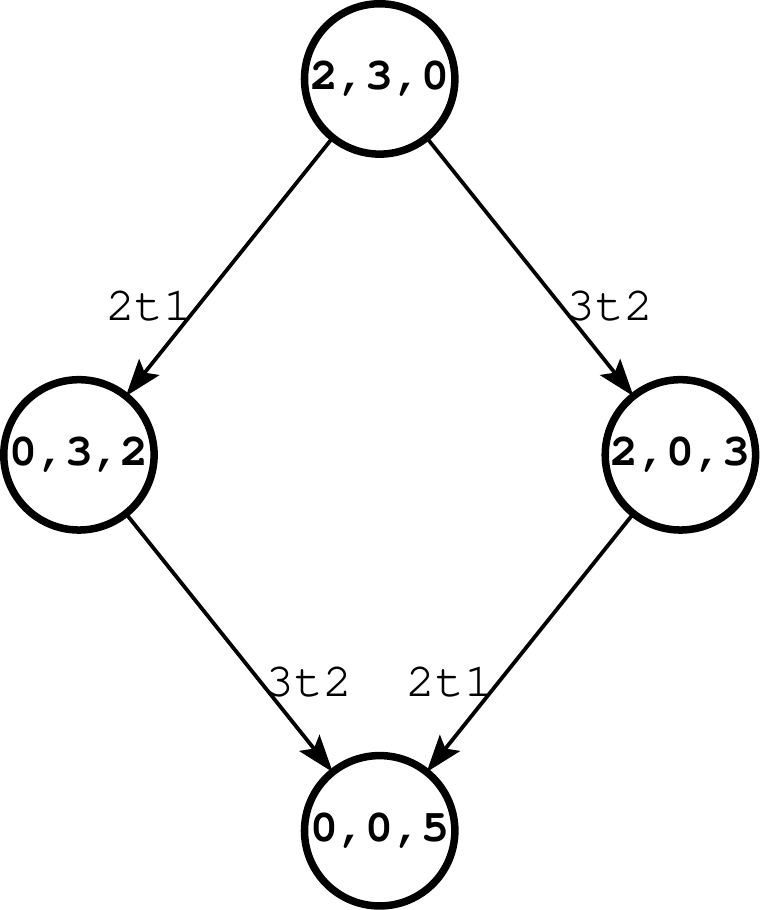}}} 
\caption{Sleptsov net.}
\label{fig-net-ex-rg-sn}
\end{subfigure}
\begin{subfigure}[h]{0.32\textwidth}\centering{\scalebox{0.32}{\includegraphics{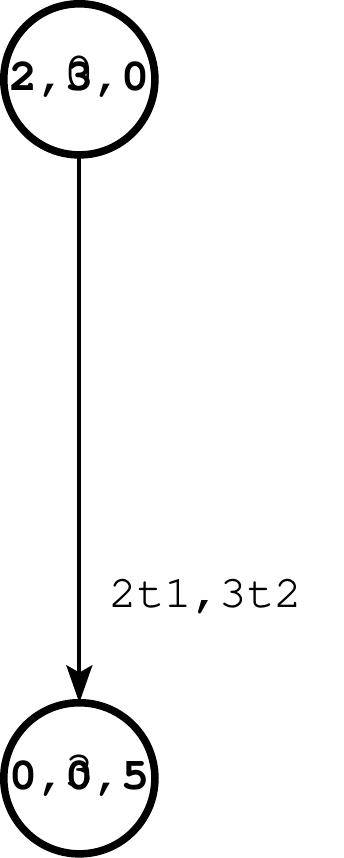}}} 
\caption{Salwicki-Sleptsov net;}
\label{fig-net-ex-rg-wsn}
\end{subfigure}
\begin{subfigure}[h]{0.32\textwidth}\centering{\scalebox{0.32}{\includegraphics{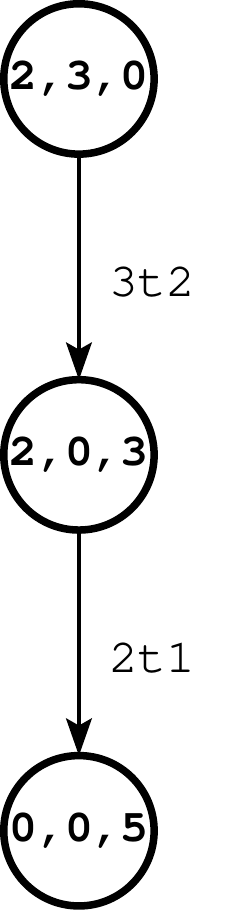}}} 
\caption{Strong Sleptsov net.}
\label{fig-net-ex-rg-snmf2}
\end{subfigure}
\begin{subfigure}[h]{0.32\textwidth}\centering{\scalebox{0.32}{\includegraphics{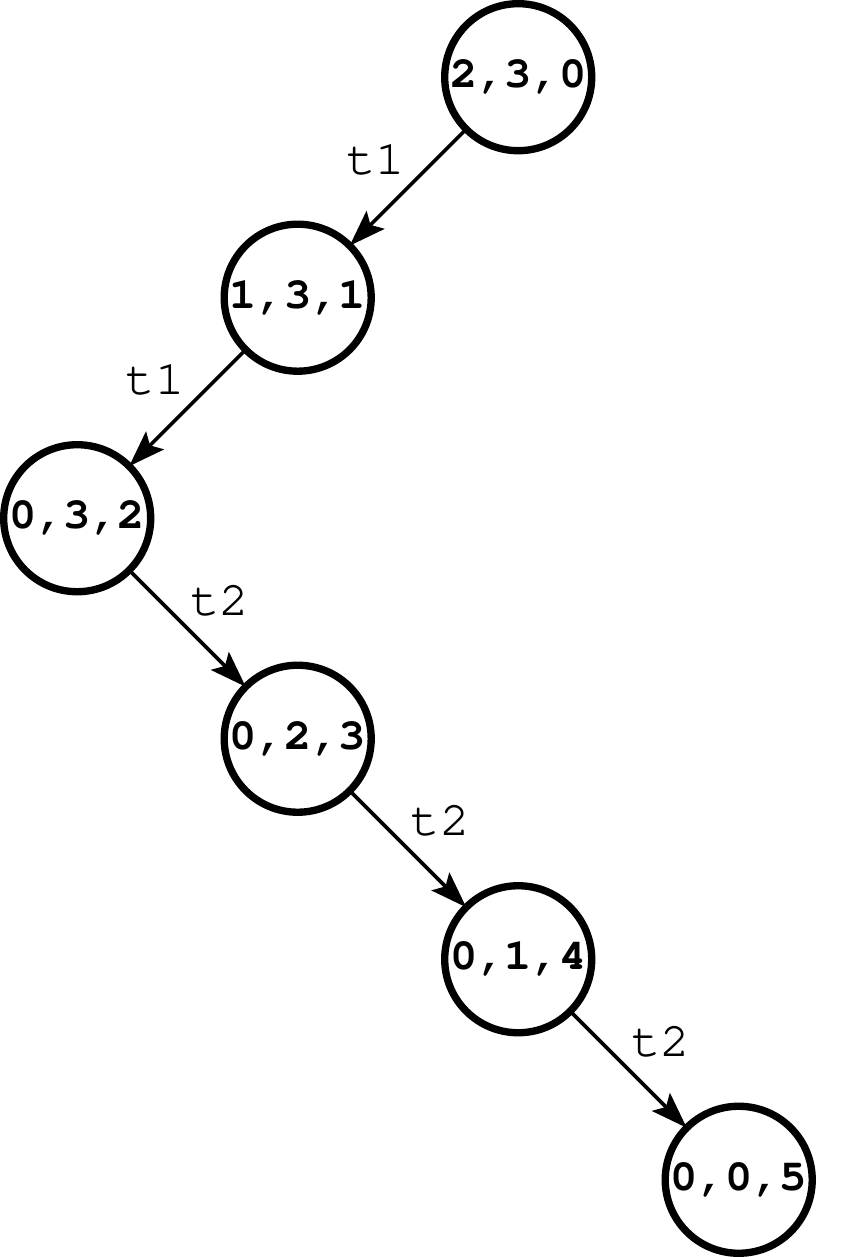}}} 
\caption{inhibitor or priority net;}
\label{fig-net-ex-rg-ippn}
\end{subfigure}
\caption{Reachability graphs of example nets in fig.~\ref{fig-net-ex}. }
\label{fig-net-ex-rg}
\end{center}
\end{figure}

Reachability graphs of a weak Salwicki net and a weak Sleptsov net for addition of two numbers (fig.~\ref{fig-net-ex-ptn}) are represented in fig.~\ref{fig-net-ex-rg-wn} and fig.~\ref{fig-net-ex-rg-sn}, respectively. For the corresponding weak Salwicki-Sleptsov net, the RG is too tangled containing an arc connecting any node with any ``next'' node, considering as the ``next'' a node having greater number of the sum (place $p_3$ marking). Weak nets with rather big freedom of behavior could be useful for solving optimization tasks where criteria are not so straightforward as the minimal number of steps. 

\begin{figure}
\begin{center}
\begin{subfigure}[h]{0.32\textwidth}\centering{\scalebox{0.32}{\includegraphics{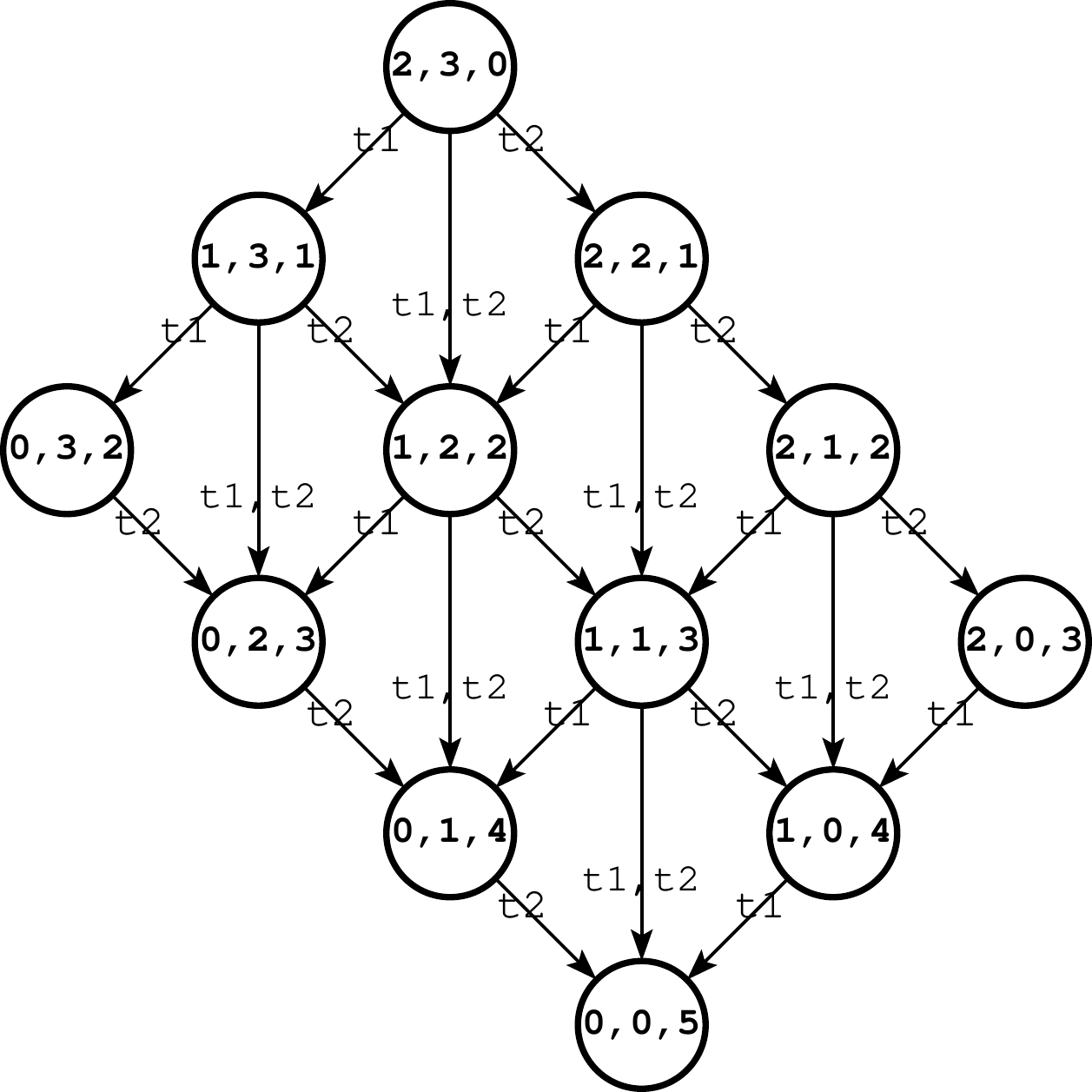}}} 
\caption{weak Salwicki net;}
\label{fig-net-ex-rg-wn}
\end{subfigure}
\begin{subfigure}[h]{0.32\textwidth}\centering{\scalebox{0.32}{\includegraphics{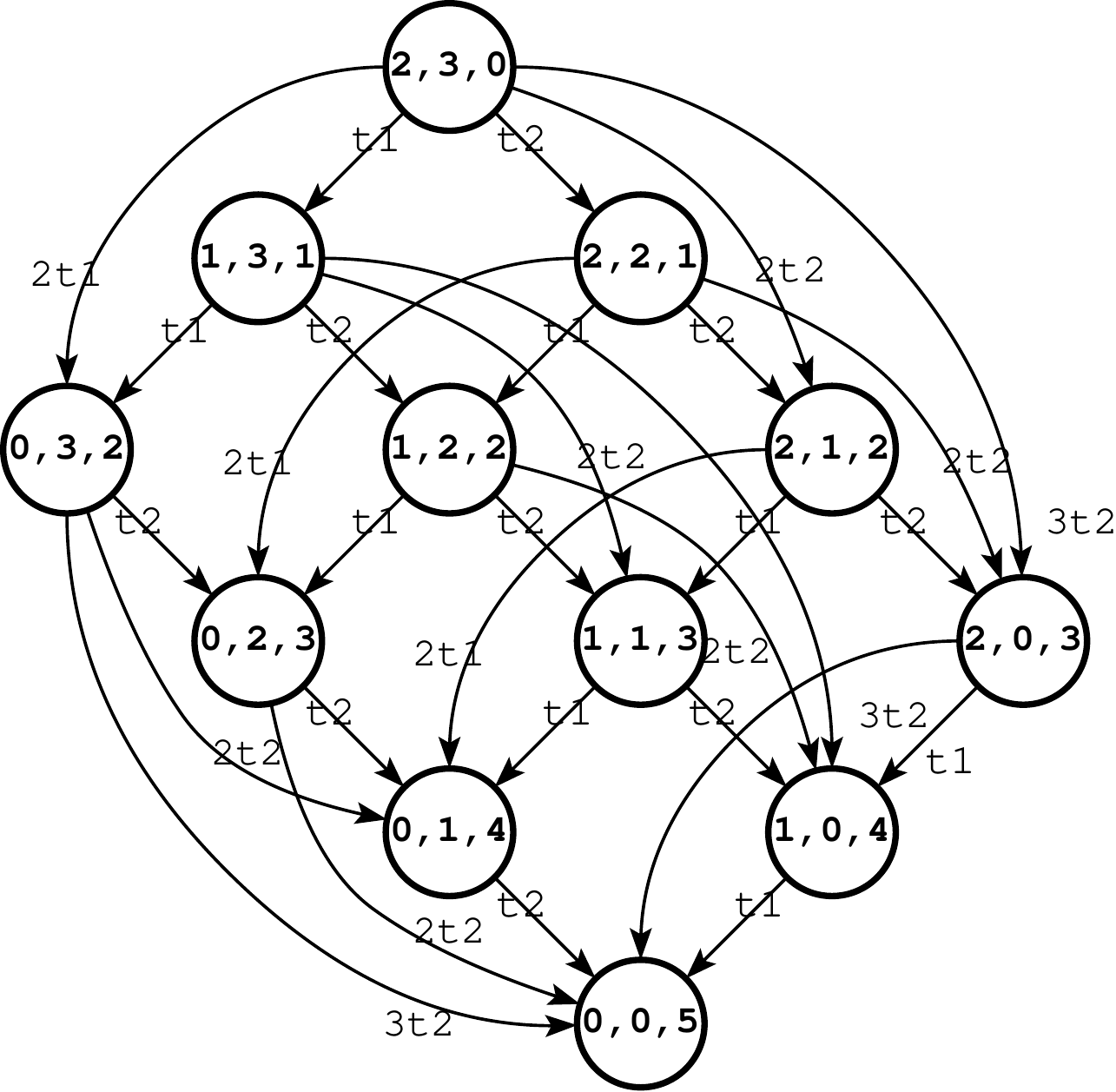}}} 
\caption{weak Sleptsov net.}
\label{fig-net-ex-rg-sn}
\end{subfigure}
\caption{Reachability graphs of weak Salwicki and Sleptsov nets.}
\label{fig-net-ex-rg}
\end{center}
\end{figure}

\section{Simulating Register Machine by Strong Sleptsov Net}
\label{sim-rm-by-ssn}

We consider a register machine (RM) as described in \cite{Shepherdson} where it has been proven that an RM is Turing-complete. Note that, an RM is rather close in its definition to a program machine of Minsky \cite{Minsky} and to a counter machine \cite{Fischer}. Here we simulate, using a technique similar to \cite{Peterson81}, an RM by a strong Sleptsov net to prove the fact that a strong Sleptsov net is Turing-complete. 

\subsection{Composition of Strong Sleptsov Net on a Given Register Machine}

A register machine represents a storage containing $n$ registers to store a nonnegative integer number in each of them and a control unit that executes a given program. We show the way how an algorithm, encoded by a register machine program, can be represented by a strong Sleptsov net. A program of an RM is a sequence of $m$ instructions over finite set of registers (variables) having nonnegative integer values. There are three following valid types of  instructions:
\begin{itemize}
\item $P(i)$: increase register $i$ ny 1.
\item $Q(i)$: decrease register $i$ by 1 (register $i$ is not zero).	
\item $J(i)[k]$: jump to instruction $k$ if resister $i$ is zero.
\end{itemize}

We represent $n$ registers, employed in a program, by $n$ places $r_1, \dots, r_n$. We simulate two first instruction of RM $P(i)$ and $Q(i)$ by PTNs shown in fig.~\ref{fig-net-inc} and fig.~\ref{fig-net-dec}, respectively. Place $p_1$ starts the instruction, place $p_2$ indicates its completion, and place $p_3$ represents a register. 

Let us consider known implementations of zero check instruction J(i) by inhibitor \cite{Peterson81}, priority \cite{Kotov84}, and Salwicki \cite{Burkhard} net shown in fig.~\ref{fig-net-zero-check}. Enumerating places, we consider $p_1$ as a start place, $p_2$ as a finish place, $p_3$ as a jump place, and $p_4$ as a register, other places and transitions are enumerated in an arbitrary order. 

Having separate models for a sequence of instructions of a given RM program, we merge (unite) finish place of the previous instruction with the start place of the next instruction to compose a model of RM program. 

\begin{figure}
\begin{center}
\begin{subfigure}[h]{0.32\textwidth}\centering{\scalebox{0.32}{\includegraphics{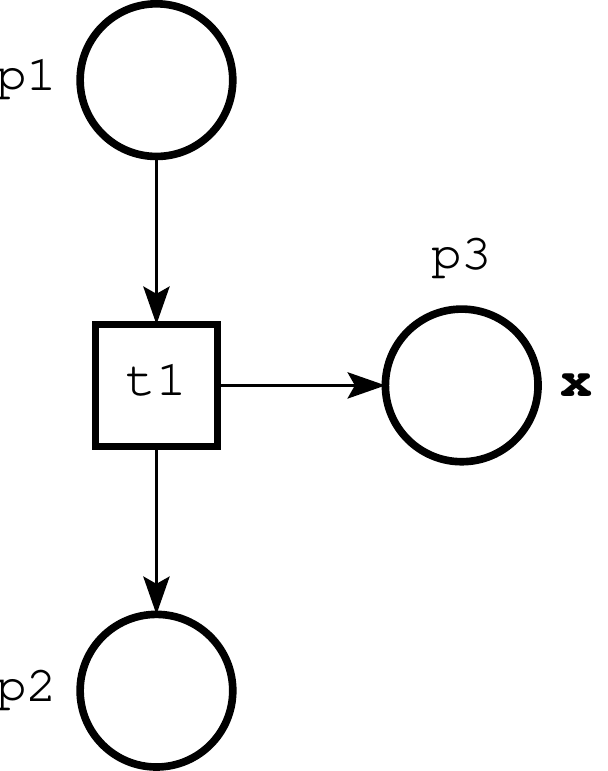}}} 
\caption{increase register $x$ by 1;}
\label{fig-net-inc}
\end{subfigure}
\begin{subfigure}[h]{0.32\textwidth}\centering{\scalebox{0.32}{\includegraphics{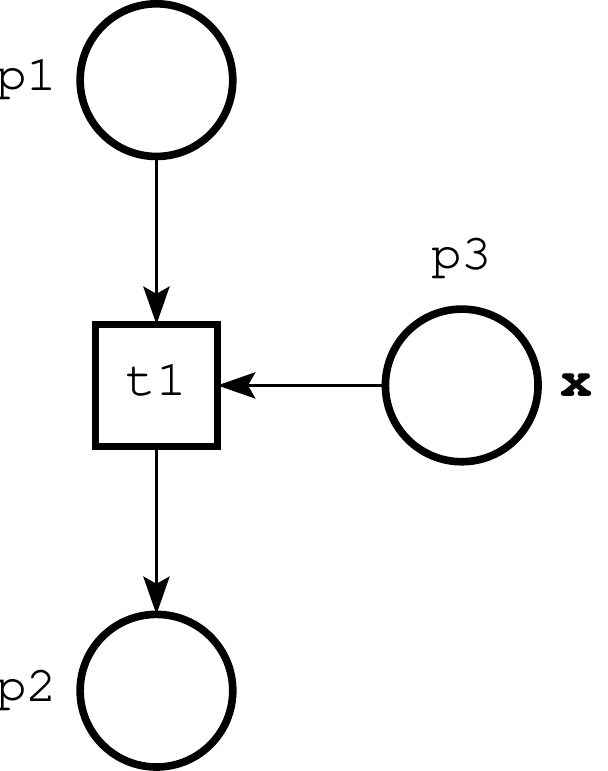}}} 
\caption{decrease register $x$ by 1.}
\label{fig-net-dec}
\end{subfigure}
\caption{Simulating increment and decrement instructions of RM.}
\label{fig-net-inc-dec}
\end{center}
\end{figure}

For an RM program consisting of $m$ instructions, we represent the control flow of program by $m+1$ places renaming the instruction $j$ start place as $q_j$ and the last  instruction $m$ finish place as $q_{m+1}$. Thus, $q_1$ serves as the program start and $q_{m+1}$ indicates that the program has been completed. A statement $j$ is started by place $q_j$ and finished by place $q_{j+1}$, the branching instruction $J(i)$ is additionally connected to place $q_s$ for the second branch, merging its place $p_3$ with the corresponding $q_s$. We renumerate internal places of all the instructions in an arbitrary order starting from $q_{m+2}$, and renumerate all the transitions sequentially. Putting a token into place $q_1$ starts the program, its completion is observed by arrival of a token into place $q_{s+1}$. To comply with the PTN definition, we specify $P=Q \cup R$. 

\begin{figure}
\begin{center}
\begin{subfigure}[h]{0.32\textwidth}\centering{\scalebox{0.32}{\includegraphics{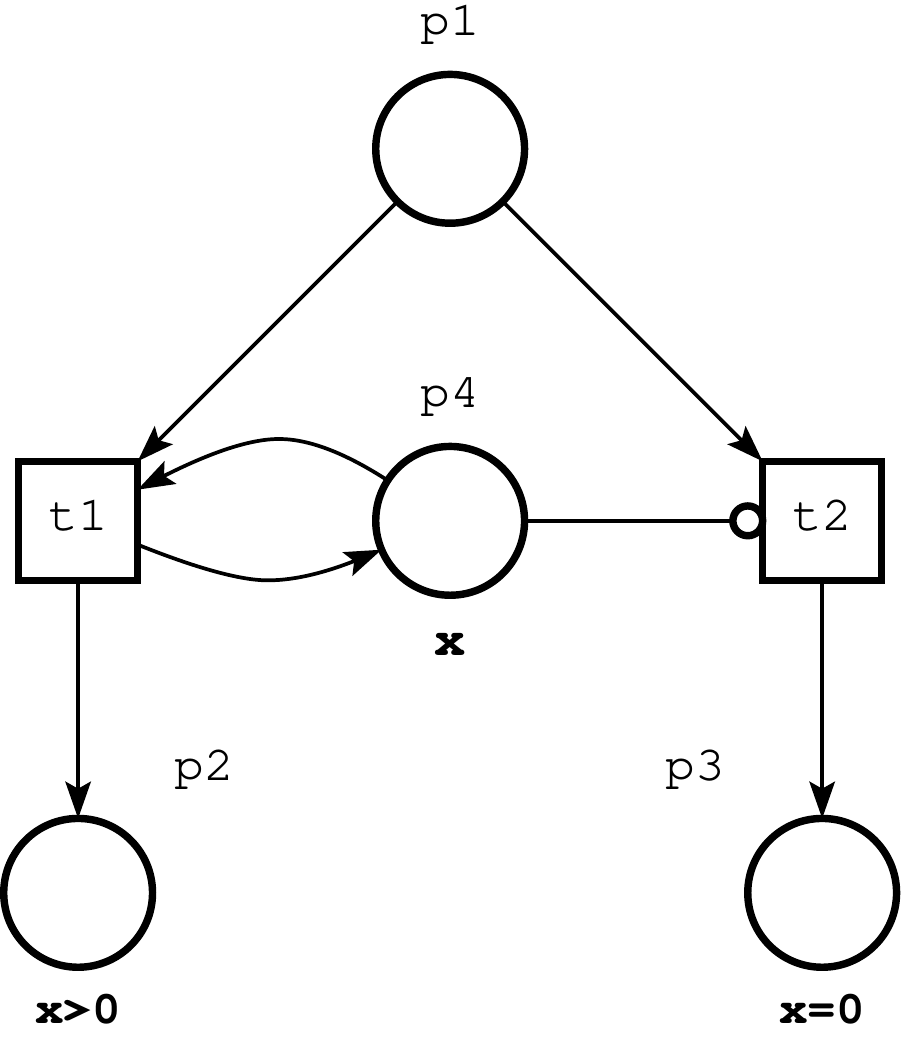}}} 
\caption{inhibitor net;}
\end{subfigure}
\begin{subfigure}[h]{0.32\textwidth}\centering{\scalebox{0.32}{\includegraphics{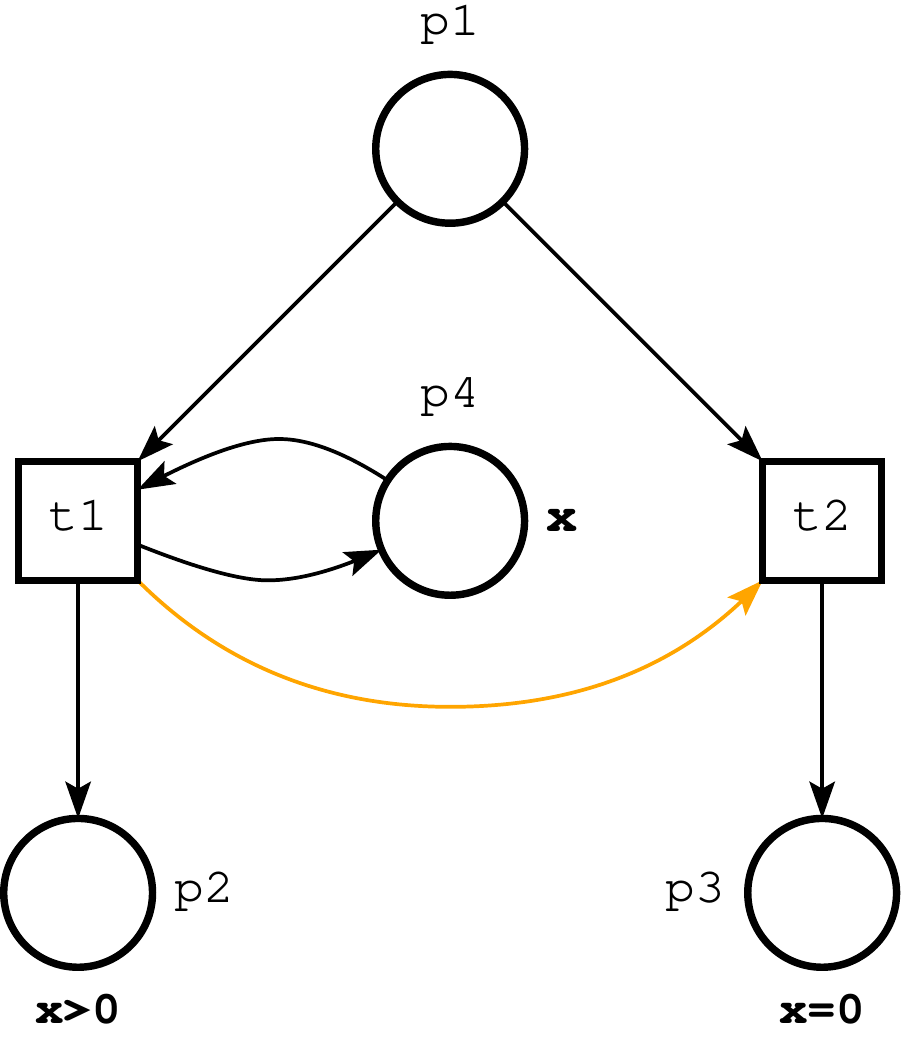}}} 
\caption{priority net;}
\end{subfigure}
\begin{subfigure}[h]{0.32\textwidth}\centering{\scalebox{0.32}{\includegraphics{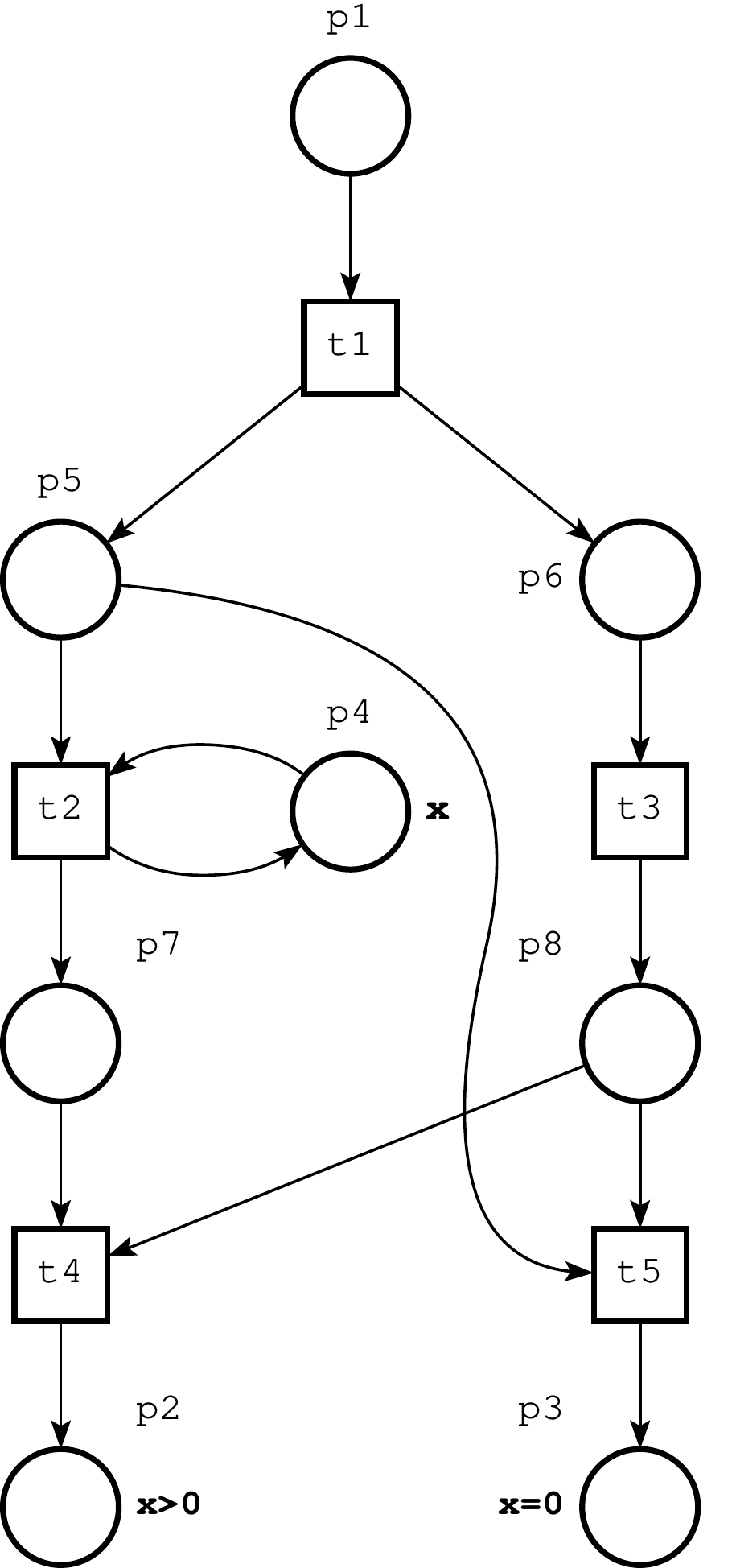}}} 
\caption{Salwicki net.}
\end{subfigure}
\caption{Implementation of zero check. }
\label{fig-net-zero-check}
\end{center}
\end{figure}

Note that we use the same simulation techniques for all the considered classes of PTNs including Sleptsov nets. Here we do not employ ability of Sleptsov nets for fast computations \cite{SNRF} having the same simple control flow simulated by the token passage from place $q_1$ to place $q_{m+1}$.

\subsection{Example of Programming in Register Machine Language}

Let us consider an example of programming RM and simulating an RM program by an inhibitor Petri net. We compose a program that computes a sum of an arithmetic progression $0,1,\dots,n$ for a given $n$. 

We do not use well-known formula for the arithmetic progression sum because it employs multiplication preferring a straightforward approach of summing up members of the progression:

\begin{lstlisting}
s=0;
for(i=1;i<n;i++) 
  s+=i;
\end{lstlisting}

Note that in RM we can not use directly addition and comparison of two numbers. We rewrite the above algorithm composing C program using increment, decrement, comparison with zero, and loop ``while'' only: 

\begin{lstlisting}
#include <stdio.h>
#include <stdlib.h>

int main(int argc, char *argv[])
{
  unsigned int s=0, n=atoi(argv[1]), nn=0;

  while(n>0)
  {
    while(n>0) // s+=n
    {
      s++;
      n--;
      nn++;
    }
    nn--;       // n--
    while(nn>0) // recover n
    {
      nn--;
      n++;
    }
  }  
  printf("%u\n",s);
}
\end{lstlisting}

The program has been debugged using gcc compiler in Linux operating system, an example of the corresponding command lines follows:

\begin{lstlisting}
>gcc -o apsum apsum.c
>./apsum 4
10
\end{lstlisting}

Observing direct correspondence of C operators and RM instructions, we compose table~\ref{C-RM-table}. Unconditional ``goto'' is represented by J(i0)[s] instruction of RM where r\_i0 is a dedicated register which value equals 0 initially and does not change. 

\begin{table}[h!]
\small
\caption{Simulating C operators by RM instructions.\label{C-RM-table}}
\begin{tabular}{|m{6cm}|m{6cm}|}
\toprule
C program & RM program \\
\midrule
\begin{lstlisting}
r_i++ 
\end{lstlisting}
& 
\begin{lstlisting}
P(i)
\end{lstlisting}\\
\hline
\begin{lstlisting}
r_i-- 
\end{lstlisting}
&
\begin{lstlisting}
Q(i)
\end{lstlisting}  \\
\hline
\begin{lstlisting}
if( r_i==0 ) goto l_s
\end{lstlisting}
& 
\begin{lstlisting}
J(i)[s]
\end{lstlisting} \\
\hline
\begin{lstlisting}
goto l_s
\end{lstlisting}
& 
\begin{lstlisting}
J(i0)[s]
\end{lstlisting} \\
\hline
\begin{lstlisting}
while( r_i>0 )
{
  // some operators
}
\end{lstlisting}
& 
\begin{lstlisting}
j:   J(i)[k+1]
     // some operators
k:   J(i0)[j]
k+1: 
\end{lstlisting}\\
\bottomrule
\end{tabular}\\[10pt]
\end{table}

The program modification with enumerated registers and instructions having operators composed according to table~\ref{C-RM-table}, supplied with the corresponding RM program in comments, follows (basic part only):  

\begin{lstlisting}
  unsigned int r_1=0 /* s */, 
               r_2=atoi(argv[1]) /* n */, 
               r_3=0 /* nn */, 
               r_4=0 /* constant zero */;
  
  l_1:  if(r_2==0) goto i13;	// J(2)[13]
  l_2:  if(r_2==0) goto i7;	// J(2)[7]
  l_3:  r_1++;			// P(1)
  l_4:  r_2--;			// Q(2)
  l_5:  r_3++;			// P(3)
  l_6:  if(r_4==0) goto i2;	// J(4)[2]
  l_7:  r_3--;			// Q(3)
  l_8:  if(r_3==0) goto i12;	// J(3)[12]
  l_9:  r_3--;			// Q(3)
  l_10: r_2++;			// P(2)
  l_11: if(r_4==0) goto i8;	// J(4)[8]
  l_12: if(r_4==0) goto i1;	// J(4)[1]
  l_13:;			//
\end{lstlisting}

The corresponding inhibitor Petri net, composed according to graphical encoding of RM operators represented in fig.~\ref{fig-net-inc-dec} and fig.~\ref{fig-net-zero-check}, is shown in fig.~\ref{fig-apsum-rm-ipn}. The net has been edited and debugged in the environment of modeling system Tina \cite{Tina}.

\begin{figure}[!t]
\center
\includegraphics[width=0.99\textwidth]{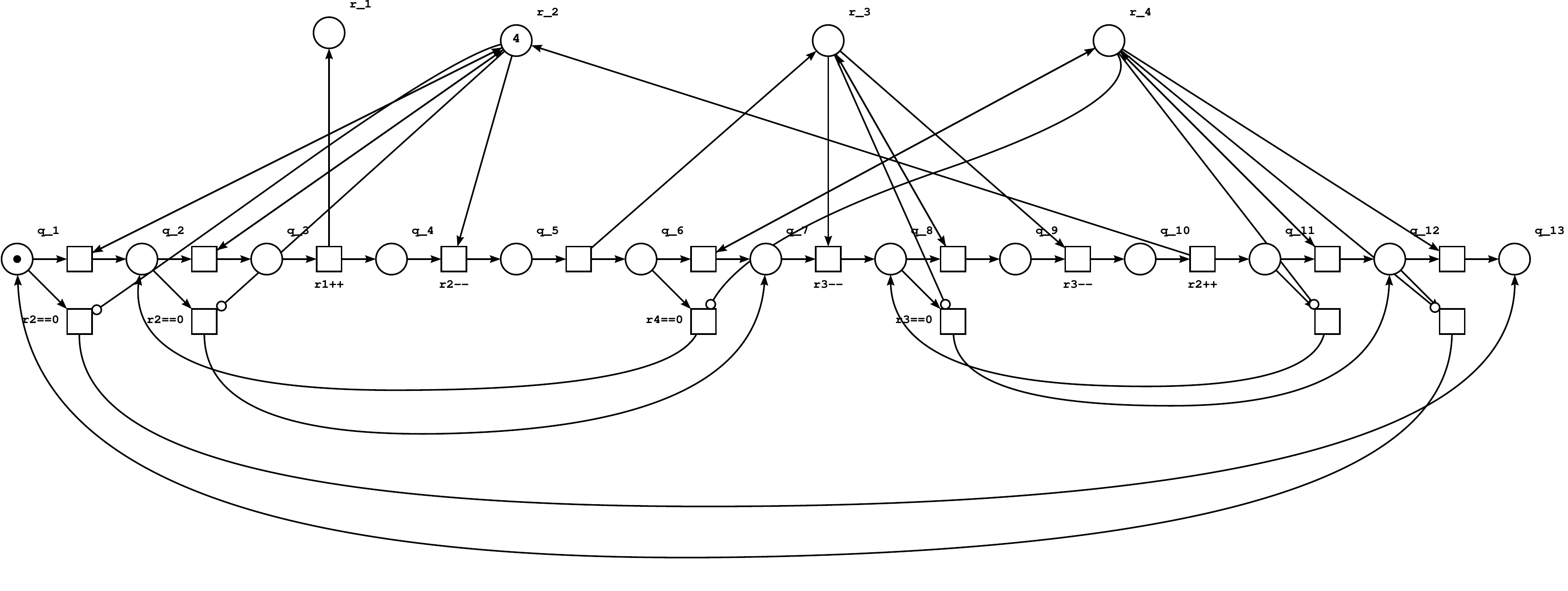}
\caption{Inhibitor Petri net simulating RM for computing sum of arithmetic progression.}
\label{fig-apsum-rm-ipn}
\end{figure}

Note that for more convenient composition of RM programs, if required, we can develop dedicated programming technology, creating library of RMs implementing basic arithmetic and logic operations and patterns to specify basic operators of branching and loop.

\section{Zero Check by Strong Sleptsov Net}

We compose a zero check strong Sleptsov net which function is the same as for other PTN classes shown in fig.~\ref{fig-net-zero-check}. The only difference consists in the fact we store, as a place marking, an incremented value of a register storing 0 as 1, 1 as 2, and so on. Zero marking serves for representing a variable with no value.

\begin{figure}[!t]
\center
\includegraphics[width=0.5\textwidth]{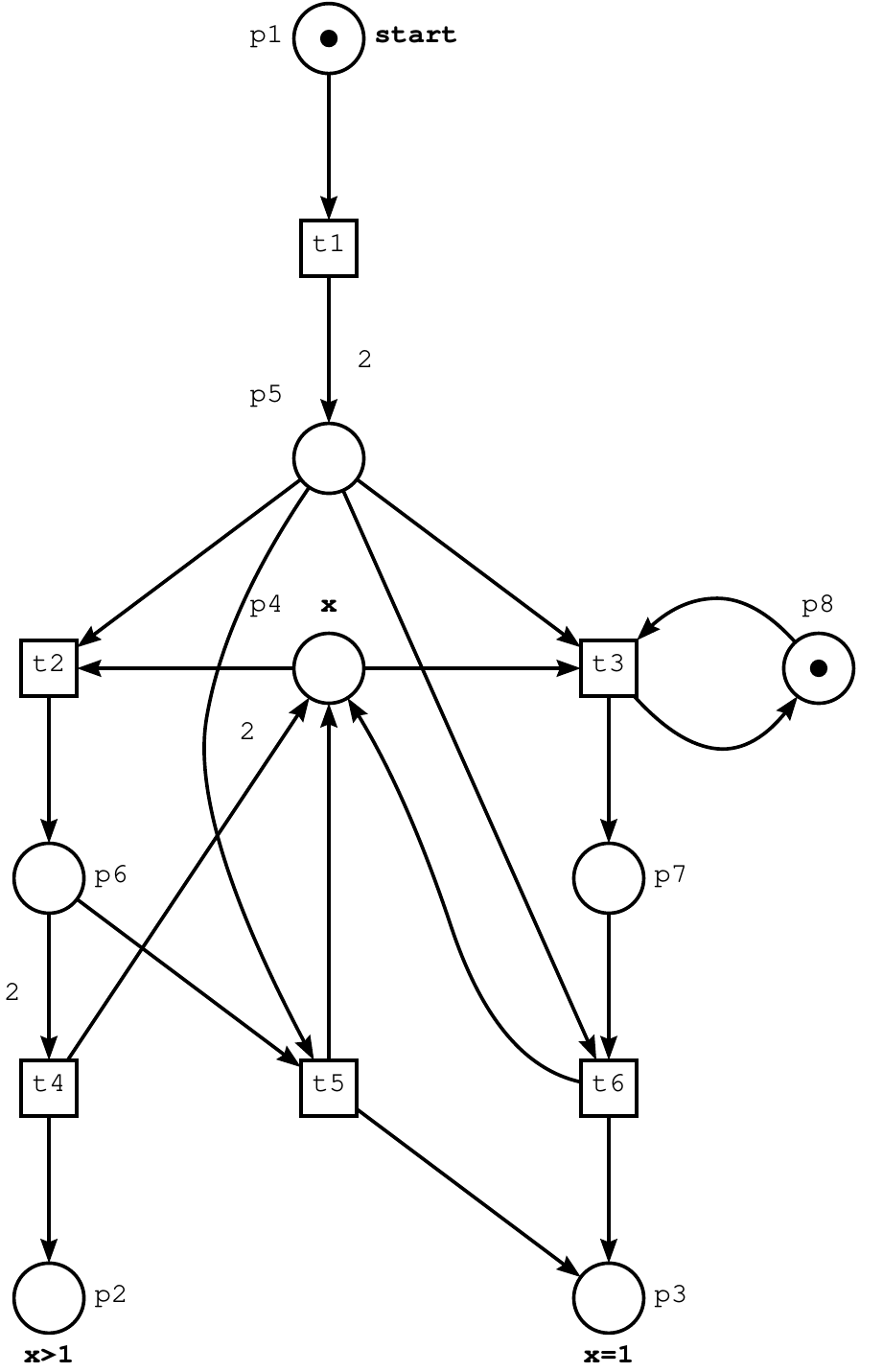}
\caption{Zero check strong Sleptsov net.}
\label{fig-sn-check-0}
\end{figure}

\begin{theorem}
Strong Sleptsov net represented in Fig.~\ref{fig-sn-check-0} implements zero check of place $x$.
\end{theorem}
\begin{proof}
We prove the theorem in a constructive way analysing and specifying permitted sequences of transitions firing. 

Firstly, transition $t_1$ fires transforming the initial marking $\{p_1,x\cdot p_4,p_8 \}$ into $\{x\cdot p_4,2\cdot p_5,p_8 \}$. Now we consider separately two basic cases when a) $x>1$ and b) $x=1$; note that when $x=0$ no transition is permitted.

For the case a), in $\{x\cdot p_4,2\cdot p_5,p_8 \}$ only one transition firing sequence $(2\cdot t_2)t_4$ is permitted. Indeed, at the first step two transitions are firable: $t_2$ with multiplicity $2$ ($min(2,x)$, $x>1$) and $t_3$ with multiplicity $1$ because of restricting place $p_8$ having marking equal to one. According to the strong Sleptsov net firing rule, $t_2$ is chosen. Then there in no alternatives. Thus, at $x>1$ we have
$$\{p_1,x\cdot p_4,p_8 \} {\rightarrow}^{t_1}$$
$$\{x\cdot p_4,2\cdot p_5,p_8 \} {\rightarrow}^{2\cdot t_2}$$
$$\{(x-2)\cdot p_4,2\cdot p_6,p_8 \} {\rightarrow}^{t_4} $$
$$\{p_2,x\cdot p_4,p_8 \}$$
which is the only transition sequence and it leads to the correct result -- marking place $p_2$ ($x>1$) preserving the rest of marking save the resetting of starting place $p_1$.

For the case b), in $\{x\cdot p_4,2\cdot p_5,x_8 \}$ two transitions $t_2$ and $t_3$ are firable with multiplicity 1 (because $x=1$). They represent an alternative -- any of them may fire because they have equal firing multiplicity. The rest of sequence for both cases has no alternatives. Thus at $x=1$ we have either
$$\{p_1,x\cdot p_4,p_8 \} {\rightarrow}^{t_1}$$
$$\{x\cdot p_4,2\cdot p_5,p_8 \} {\rightarrow}^{t_2}$$
$$\{p_5,p_6,p_8 \} {\rightarrow}^{t_5} $$
$$\{p_3,x\cdot p_4,p_8 \}$$
or
$$\{p_1,x\cdot p_4,x_8 \} {\rightarrow}^{t_1}$$
$$\{x\cdot p_4,2\cdot p_5,p_8 \} {\rightarrow}^{t_3}$$
$$\{p_5,p_7,p_8 \} {\rightarrow}^{t_6} $$
$$\{p_3,x\cdot p_4,p_8 \}.$$

Two alternative sequences lead to the same result -- marking place $p_3$ ($x=1$) preserving the rest of marking save the resetting of starting place $p_1$. The corresponding parametric RG is shown in fig.~\ref{fig-rg-sn-check-0}. 
\end{proof}

\begin{figure}[!t]
\center
\includegraphics[width=0.6\textwidth]{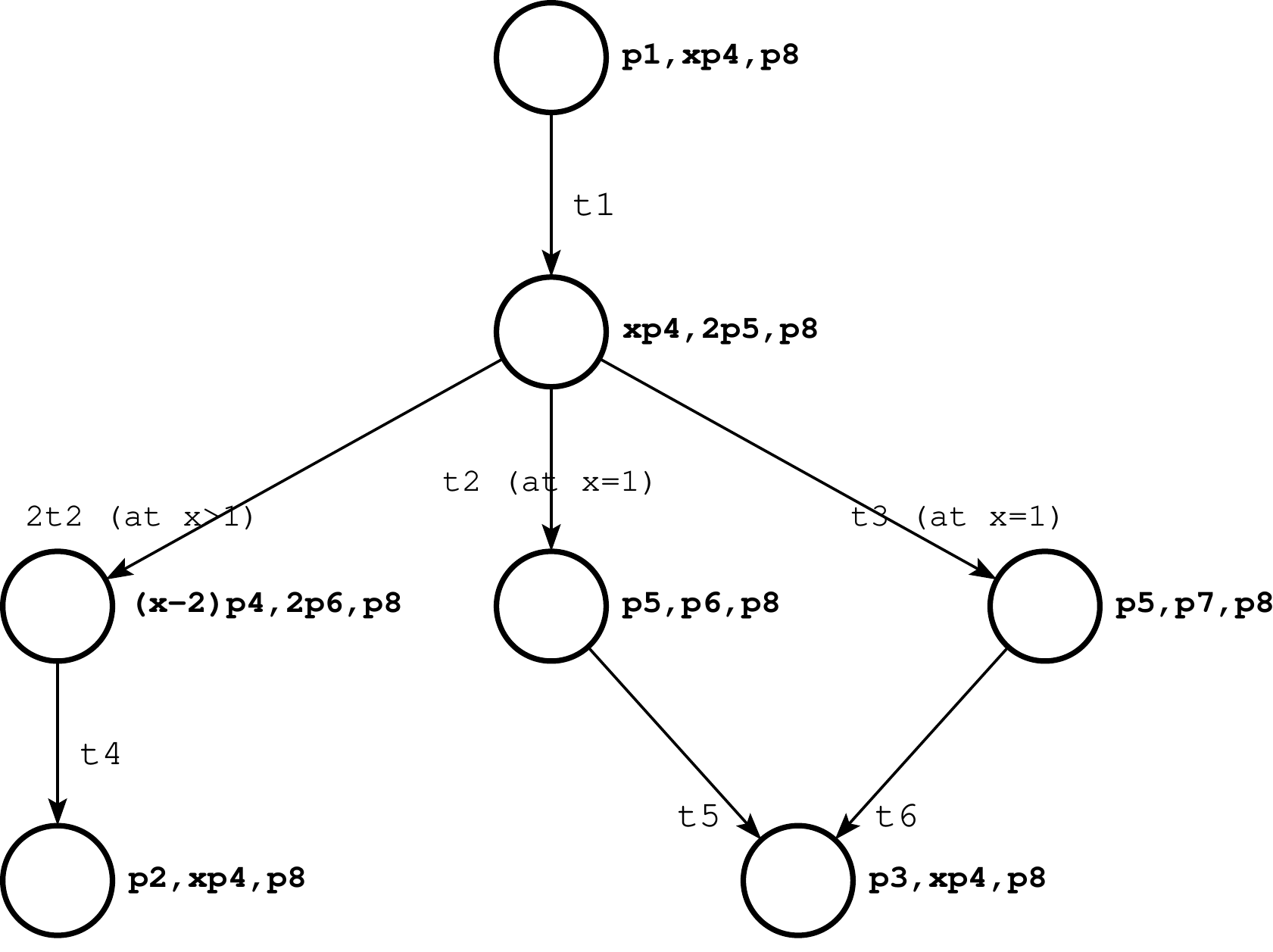}
\caption{Parametric reachability graph for zero check strong Sleptsov net (fig.~\ref{fig-sn-check-0}).}
\label{fig-rg-sn-check-0}
\end{figure}

As immediately inferred from the above theorem, constructs of Section~\ref{sim-rm-by-ssn}, and \cite{Shepherdson}, we formulate the following corollary.

\begin{corollary}
A strong Sleptsov net is Turing-complete.
\end{corollary}

\section{Conclusions}

A strong Sleptsov net, where a transition with the maximal firing multiplicity fires at a step, has been introduced, its Turing-completeness proven. It allows us to compose Sleptsov net programs without using additional concepts such as inhibitor arcs or transition priorities. 

A classification of place-transition nets with regard to transition firing rules, based on general definitions and their weak and strong variants, has been presented.

A question whether a general Sleptsov net is Turing-complete remains open representing a direction for future research.


\begin{thebibliography}{00}

\bibitem{ZaitsevPhD}
D.A. Zaitsev, Solving operative management tasks of a discrete manufacture via Petri net models, PhD thesis, Kiev, the Academy of sciences of Ukraine, Institute of Cybernetics name of V.M.Glushkov, 1991. In Russ. \url{https://daze.ho.ua/daze-phd-1991.pdf}

\bibitem{CiSA97}
D.A. Zaitsev, A.I. Sleptsov, State equations and equivalent transformations for timed petri nets. Cybern. Syst. Anal. 33 (1997) 659--672. \url{https://doi.org/10.1007/BF02667189}

\bibitem{Neary}
T. Neary, On the computational complexity of spiking neural P systems, Nat. Comput. 9(4)  (2010) 831--851.

\bibitem{ExUoR1}
X. Zhang, Y. Jiang, L. Pan, Small universal spiking neural P systems with exhaustive use of rules, 3rd International Conference on Bio-Inspired Computing: Theories and Applications (2008) 117--128. \url{https://doi.org/10.1109/BICTA.2008.4656713}

\bibitem{ExUoR2}
G. Paun, M. Ionescu, T. Yokomori, Spiking Neural P Systems with an Exhaustive Use of Rules. Int. J. Unconv. Comput. 3 (2007) 135--153.

\bibitem{DNAcomp}
P. Formanowicz, DNA computing, Computational Methods in Science and  Technology 11(1)  (2005) 11--20. \url{https://doi.org/10.12921/cmst.2005.11.01.11-20}

\bibitem{Winfree1}
L. Qian, E Winfree, Scaling up digital circuit computation with DNA strand displacement cascades, Science 332(6034) (2011) 1196--1201. \url{https://doi.org/10.1126/science.1200520}

\bibitem{Winfree2}
L. Qian, E. Winfree, A simple DNA gate motif for synthesizing large-scale circuits
Journal of The Royal Society Interface 8(62) (2011) 1281--1297. \url{https://doi.org/10.1098/rsif.2010.0729}

\bibitem{Winfree3}
D. Woods, H.-L. Chen, S. Goodfriend, N. Dabby, E. Winfree, P. Yin, Active self-assembly of algorithmic shapes and patterns in polylogarithmic time, ITCS '13: Proceedings of the 4th conference on Innovations in Theoretical Computer Science (2013) 353--354. https://doi.org/10.1145/2422436.2422476

\bibitem{Alhazov}
A. Alhazov, S. Verlan, Minimization strategies for maximally parallel multiset rewriting systems, Theoretical Computer Science 412(17) (2011) 1581--1591. \url{https://doi.org/10.1016/j.tcs.2010.10.033}

\bibitem{SNRF}
D.A. Zaitsev, Sleptsov Nets Run Fast, IEEE Transactions on Systems, Man, and Cybernetics: Systems 46(5) 682--693. \url{https://doi.org/10.1109/TSMC.2015.2444414}

\bibitem{Petri62}
C.A. Petri, Kommunication mit Automaten, Technischen Hoschule Darmstadt, 1962, Ph.D. thesis.

\bibitem{Peterson81}
J.L. Peterson, Petri Net Theory and the Modeling of Systems, Englewood Cliffs, N.J.: Prentice-Hall, 1981.

\bibitem{Burkhard}
H.-D. Burkhard, On priorities of parallelism: Petri nets under the
maximum firing strategy, in Salwicki A. (ed.) Logics of Programs and Their Applications. Lecture Notes in Computer Science, vol. 148. Berlin, Germany:
Springer, 1983, 86--97.

\bibitem{SNC}
D.A. Zaitsev,  Sleptsov Net Computing, Chapter 672 in Mehdi Khosrow-Pour (Ed.) Encyclopedia of Information Science and Technology, Fourth Edition (10 Volumes). IGI-Global: USA, 2017, 7731--7743. \url{https://doi.org/10.4018/978-1-5225-2255-3.ch672}

\bibitem{Tilak}
T. Agerwala, A complete model for representing the coordination of asynchronous processes, Hopkins Computer Science Program, Res. Rep., No. 32, John Hopkins University, Baltimore, 1974.

\bibitem{Hack74}
M. Hack, Decision problems for Petri nets and vector addition systems,
Comput. Struct. Group, Massachusetts Inst. Technol., Cambridge, MA,
USA, Tech. Rep. 95, Mar. 1974.

\bibitem{Hack76}
M. Hack, Petri Net Language, Massachusetts Institute of Technology,  Technical Report,  1976.

\bibitem{PSNL}
D. Zaitsev, J. Jürjens, Programming in the Sleptsov net language for systems control. Advances in Mechanical Engineering 8(4) (2016). \url{https://doi.org/10.1177/1687814016640159}

\bibitem{USN}
D.A. Zaitsev, Universal Sleptsov net, International Journal of Computer Mathematics, 94(12) (2017) 2396--2408. \url{https://doi.org/10.1080/00207160.2017.1283410}

\bibitem{UIPN}
D.A. Zaitsev, Universality in Infinite Petri Nets, in Durand-Lose J., Nagy B. (eds) Machines, Computations, and Universality (MCU 2015), Lecture Notes in Computer Science, vol 9288, 2015, Springer, Cham. \url{https://doi.org/10.1007/978-3-319-23111-2_12}

\bibitem{Shannon}
C.E. Shannon, A universal Turing machine with two internal states, in
Automata Studies, Annals of Math. St., Princeton, NJ: Princeton Univ.
Press, 1956, 157--165.

\bibitem{Korec}
I. Korec, Small universal register machines, Theor. Comput. Sci.
168(2) (1996) 267--301.

\bibitem{UPN}
D.A. Zaitsev, Universal petri net, Cybern. Syst. Anal. 48 (2012) 498--511. \url{https://doi.org/10.1007/s10559-012-9429-4}

\bibitem{TMUPN}
D.A. Zaitsev, Toward the Minimal Universal Petri Net, IEEE Transactions on Systems, Man, and Cybernetics: Systems 44(1) (2014) 47--58. \url{https://doi.org/10.1109/TSMC.2012.2237549}

\bibitem{SCAIPN}
D.A. Zaitsev, Simulating Cellular Automata by Infinite Petri Nets, Journal of Cellular Automata 13(1-2) (2018) 121--144.

\bibitem{SUDPN}
A. Alhazov, S. Ivanov, E. Pelz, S. Verlan, Small Universal Deterministic Petri Nets with Inhibitor Arcs, Journal of Automata, Languages and Combinatorics 21(1-2) (2016) 7--26.

\bibitem{Turing}
A.M. Turing, On computable numbers with an application to the
Entscheidungsproblem, Proc. Lond. Math. Soc. s2-42(1) (1937) 230--265.

\bibitem{Shepherdson}
J.C. Shepherdson, H.E. Sturgis, Computability of recursive
functions, J. ACM 10(2) (1963) 217--255.

\bibitem{Minsky}
M.L. Minsky, Computation: Finite and Infinite Machines, Englewood
Cliffs, NJ: Prentice-Hall, 1967.

\bibitem{ASMEflow}
A.S.M.E. standard. Operation and flow process charts, New York: The American Society of Mechanical Engineers, 1947.

\bibitem{Gilbreth}
F.B. Gilbreth, L.M. Gilbreth, Process Charts, American Society Of Mechanical Engineers, New York, 1921. https://lccn.loc.gov/ca22000471

\bibitem{FlowChart}
H.H. Goldstein and J.von Neumann, Planning and Coding Problems for an Electronic Computing Instrument, part II, in von Neumann Collected Works, Vol.V, McMillan, New York, 1947, 80--151.

\bibitem{Yanov1}
Yu.I. Yanov,  On the equivalence and transformation of program schemes, Doklady, A.N. USSR 113 (1957) 39--42; Trans., Comm. ACM I(10) (Oct. 1958) 8--12.

\bibitem{Yanov2}
Yu.I. Yanov, On matrix program schemes, Doklady A.N. USSR 118 (1957) 283--286; Trans., Comm. ACM 1(12) (Dec. 1958), 3--6. 

\bibitem{Gill}
S. Gill, Parallel Programming, The Computer Journal 1(1) (1958) 2--10. https://doi.org/10.1093/comjnl/1.1.2

\bibitem{UML}
G. Booch, J. Rumbaugh, I. Jacobson, The Unified Modeling Language User Guide, Addison-Wesley, 2005. 

\bibitem{multiset1}
M.W. Bunder, A Foundation for Multiset Theory, University of Wollongong, 1987. 

\bibitem{multiset2}
W.D. Blizard, The development of multiset theory, Modern Logic 1 (1991), 319--352.

\bibitem{SEPNC}
D.A. Zaitsev, M.C. Zhou, From strong to exact Petri net computers, International Journal of Parallel, Emergent and Distributed Systems, (2021). \url{https://doi.org/10.1080/17445760.2021.1991340}

\bibitem{Kotov84}
V.Ye. Kotov, Petri Nets, Nauka, Moscow (1984). In Russ.

\bibitem{Fischer}
P.C. Fischer, A.R. Meyer, A.L. Rosenberg, Counter machines and
counter languages, J. Math. Syst. Theory 2(3) (1968) 265--283.

\bibitem{Tina}
B. Berthomieu, P.-O. Ribet, F. Vernadat, The tool TINA -- Construction of Abstract State Spaces for Petri Nets and Time Petri Nets, International Journal of Production Research, 42(14) (July 2004).


\end{thebibliography}
\end{document}